\title{The Garden Hose Complexity for the Equality Function}
\author{
        Well Y. Chiu  \\
        Department of Applied Mathematics\\
        National Chiao Tung University\\
        1001 University Road, Hsinchu 300, Taiwan\\
        well.am94g@nctu.edu.tw
            \and
        Mario Szegedy\\
        Department of Computer Science\\
        Rutgers, the State University of New Jersey\\
        110 Frelinghuysen Road, Piscataway, NJ 08854, USA\\
        szegedy@cs.rutgers.edu
            \and
        Chengu Wang \\
        Institute for Interdisciplinary Information Sciences \\
        Tsinghua University, Beijing 100084, China \\
        wangchengu@gmail.com
            \and
        Yixin Xu \\
        Department of Computer Science\\
        Rutgers, the State University of New Jersey\\
        110 Frelinghuysen Road, Piscataway, NJ 08854, USA\\
        yixinxu@cs.rutgers.edu
}
\date{\today}

\documentclass[12pt,a4paper]{article}
\usepackage[latin1]{inputenc}
\usepackage{graphicx}
\usepackage{amsmath}
\usepackage{amsthm}
\usepackage{amsfonts}
\usepackage{amssymb}
\usepackage{hyperref}
\usepackage{braket} 
\usepackage{xcolor}
\usepackage{listings}
\usepackage{todonotes}

\hypersetup{
    colorlinks=false,
}
\newtheorem{theorem}{Theorem}

\newtheorem{example}[theorem]{Example}
\newtheorem{lemma}[theorem]{Lemma}
\newtheorem{conjecture}[theorem]{Conjecture}
\newtheorem{corollary}[theorem]{Corollary}

\newcommand{\EQ}{\text{\textsc{EQ}}}
\newcommand{\GH}{\operatorname{GH}}
\newcommand{\waterin}{\operatorname{in}}
\newcommand{\waterout}{\operatorname{out}}

\begin{document}
\maketitle

\begin{abstract}
The garden hose complexity is a new communication complexity introduced by H. Buhrman, 
S. Fehr, C. Schaffner and F. Speelman \cite{BFSS13} to analyze position-based cryptography protocols 
in the quantum setting. We focus on the garden hose complexity of the equality function, 
and improve on the bounds of O. Margalit and A. Matsliah\cite{MM12} with the help of a new approach and of 
our handmade simulated annealing based solver. We have also found beautiful symmetries of the solutions 
that have lead us to develop the notion of \textit{garden hose permutation groups}. Then, exploiting 
this new concept, we get even further, although several interesting open problems remain. 
\end{abstract}


\section{Introduction}

\subsection{Quantum position based cryptography and the garden hose model}

\textit{Position based cryptography} was first introduced in \cite{CGMO09}, while its quantum setting was introduced in \cite{Position11} and also \cite{Position11Imp}. 
The basic idea is to use the geographical location as its only credential. For example, one message might be decrypted only if the receiver is in a specified location. 
In the setting of \textit{position verification}, a special application of position based cryptography, Alice wants to convince Bob that she is in a particular position. In \cite{CGMO09} it has been shown that 
position verification using classical protocols is impossible against colluding adversaries (who control all positions except the prover's claimed position). In the quantum
setting \cite{Position11}, a general impossibility has been shown: using an enormous amount of quantum entanglement, colluding adversaries are always able to make it look to 
the verifiers as if they were at the claimed position. \\






In \cite{BFSS13}, the authors proposed a protocol $PV_{qubit}^{f}$ for one dimensional quantum position verification, while the basic ideas generalize to higher dimensions. There are two verifiers $V_0, V_1$ and one prover $P$ 
in between them and $f$ is a fixed publicly known Boolean function $f:\{0,1\}^n \times \{0,1\}^n \rightarrow \{0,1\}$. Again, the above protocol 
$PV_{qubit}^{f}$ is insecure under two adversarial attackers Alice and Bob if they share enough number of EPR pairs. It has been shown that there is a one-one correspondence 
between attacking the position-verification scheme $PV_{qubit}^{f}$ and computing the function $f$ in the \textit{garden-hose model}\cite{BFSS13}. More generally, we can 
translate any strategy of Alice and Bob in the garden-hose model to a perfect quantum attack of $PV_{qubit}^{f}$ by using one EPR pair per pipe and performing Bell 
measurements where the players connect the pipes. We omit the details of the connection between them here while emphasizing on the new communication complexity model: the garden
hose model.\\

Alice and Bob want to compute $f(x,y)$ of a Boolean function $f:\{0,1\}^n \times \{0,1\}^n \rightarrow \{0,1\}$ together where Alice gets $x$ and Bob gets $y$. They have $m$ water
pipes numbered by $1,2,\cdots,m$ between them and in addition Alice has a water tap $0$. When Alice gets input $x$, she makes a \textit{configuration} $A(x)$: 
a (nontrivial) partial matching in $\{0,1,\cdots,m\}$, i.e. Alice uses hoses to connect those pairs of pipes according to the matching. When Bob gets input $y$, 
he makes a configuration $B(y)$: a (nontrivial) partial matching in $\{1,\cdots,m\}$. Then Alice opens the water tap, when water comes out from Alice's side, we say 
the output is $1$, otherwise it is $0$. Say that $f:\{0,1\}^n \times \{0,1\}^n \rightarrow \{0,1\}$ can be computed by $m$ pipes in the garden hose model if for 
every possible input $(x,y)$ there is a configuration pair $(A(x), B(y))$ such that water comes out in correct side. The \textit{garden hose complexity} of $f$, 
denoted by $\GH(f)$, is defined as the minimum number of pipes that computes $f$. \\

We will focus on computing $\GH(\EQ_n)$ for equality function $\EQ_n$ in this paper, 

\[ \EQ_n(x,y) = \left\{
  \begin{array}{l l}
    1 & \quad \text{if $x=y$ }\\
    0 & \quad \text{if $x \neq y$.}
  \end{array} \right.\]

\begin{example}
Let $n = 1$. Then three pipes suffice for Alice and Bob in computing $\EQ_1$, i.e. $\GH(\EQ_1) \leq 3$. Here is one solution:
$$A(0)=\{01\},\quad A(1)=\{02\},$$ $$B(0)=\{13\}, \quad B(1)=\{23\}.$$ The following figure describes the configuration pictorially. 
\end{example}

\begin{figure}[htb]
\begin{center}
\includegraphics[scale=0.7]{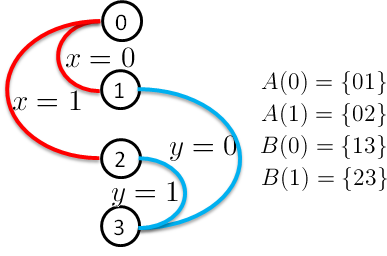}
\caption{Implementation of $\EQ_1$ function}
\end{center}
\end{figure}

\subsection{Prior results and our results}

It has been known that $\GH(\EQ_n)$ is lower bounded by $n+1$. Actually computing upper bound for $\GH(\EQ_n)$ featured as April 2012's ``Ponder This'' puzzle on the IBM 
website\footnote{\url{http://ibm.co/I7yvMz}}. The best solution there gives $$\GH(\EQ_n) \leq \frac{8}{\log 46} \cdot n + O(1) \approx 1.448 n + O(1)$$ by applying IBM SAT-Solver\cite{MM12}. With the help of a 
new approach and of our handmade simulated annealing based solver, we have improved their bounds. 
We have also found beautiful symmetries of the solutions 
that have lead us to develop the notion of \textit{garden hose permutation groups}. Then, exploiting 
this new concept, we push the upper bound to the following. 

\begin{theorem} \label{mainthm}
The garden hose complexity of the equality function: $$\GH(\EQ_n) \leq \frac{28}{\log 3^{13}}\cdot n + O(1) \approx 1.359 n + O(1). $$ 
\end{theorem}

\section{Matrix idea}

Similar to the role of matrix in communication complexity, here we introduce the \textit{configuration matrix} $M_m$ for garden hose model where $m$ is the number of water pipes. 
Rows(columns) in $M_m$ indicate all possible configurations for Alice(Bob). Each entry is either 0 or 1 determined by which direction the water comes out 
according to its corresponding row and column. Below explicitly describes all elements in $M_3$ and $M_4$, two smallest nontrivial configuration matrices.

$$ M_3 = \begin{matrix}
 & \begin{matrix}
\text{ 12 } & \text{13 } & \text{23 }
\end{matrix} \\ 
\begin{matrix}
01\\ 
02\\ 
03
\end{matrix} &  \begin{bmatrix}
\text{ 1 } & \textbf{ \textcolor{red}{1} } & \textbf{ \textcolor{red}{0} }\\ 
\text{ 1 } & \textbf{ \textcolor{red}{0} } & \textbf{ \textcolor{red}{1} }\\ 
\text{ 0 } & \text{ 1 } & \text{ 1 }
\end{bmatrix}
\end{matrix}$$

$$M_4 = \begin{matrix}
 & \begin{matrix}
\text{ 12 } & \text{13 } & \text{14 } & \text{23 } & \text{24 } & \text{34 }
\end{matrix} \\ 
\begin{matrix}
01\\ 
02\\ 
03\\
04\\
01,23\\
01,24\\
01,34\\
02,13\\
02,14\\
02,34\\
03,12\\
03,14\\
03,24\\
04,12\\
04,13\\
04,23
\end{matrix} &  \begin{bmatrix}
\text{ 1 } & \text{ 1 } & \text{ 1 } & \text{ 0 } & \text{ 0 } & \text{ 0 }\\ 
\text{ 1 } & \text{ 0 } & \text{ 0 } & \text{ 1 } & \text{ 1 } & \text{ 0 }\\ 
\text{ 0 } & \text{ 1 } & \text{ 0 } & \text{ 1 } & \text{ 0 } & \text{ 1 }\\ 
\text{ 0 } & \text{ 0 } & \text{ 1 } & \text{ 0 } & \text{ 1 } & \text{ 1 }\\ 
\textbf{ 0 } & \textbf{ 0 } & \textbf{ \textcolor{red}{1} } & \textbf{ 0 } & \textbf{ 0 } & \textbf{ 0 }\\ 
\textbf{ 0 } & \textbf{ \textcolor{red}{1} } & \textbf{ 0 } & \textbf{ 0 } & \textbf{ 0 } & \textbf{ 0 }\\ 
\textbf{ \textcolor{red}{1} } & \textbf{ 0 } & \textbf{ 0 } & \textbf{ 0 } & \textbf{ 0 } & \textbf{ 0 }\\ 
\textbf{ 0 } & \textbf{ 0 } & \textbf{ 0 } & \textbf{ 0 } & \textbf{ \textcolor{red}{1} } & \textbf{ 0 }\\ 
\textbf{ 0 } & \textbf{ 0 } & \textbf{ 0 } & \textbf{ \textcolor{red}{1} } & \textbf{ 0 } & \textbf{ 0 }\\  
\text{ 1 } & \text{ 0 } & \text{ 0 } & \text{ 0 } & \text{ 0 } & \text{ 0 }\\ 
\textbf{ 0 } & \textbf{ 0 } & \textbf{ 0 } & \textbf{ 0 } & \textbf{ 0 } & \textbf{ \textcolor{red}{1} }\\ 
\text{ 0 } & \text{ 0 } & \text{ 0 } & \text{ 1 } & \text{ 0 } & \text{ 0 }\\ 
\text{ 0 } & \text{ 1 } & \text{ 0 } & \text{ 0 } & \text{ 0 } & \text{ 0 }\\ 
\text{ 0 } & \text{ 0 } & \text{ 0 } & \text{ 0 } & \text{ 0 } & \text{ 1 }\\ 
\text{ 0 } & \text{ 0 } & \text{ 0 } & \text{ 0 } & \text{ 1 } & \text{ 0 }\\
\text{ 0 } & \text{ 0 } & \text{ 1 } & \text{ 0 } & \text{ 0 } & \text{ 0 }
\end{bmatrix}
\end{matrix}$$


We observe that: 

\begin{lemma}
 $\GH(\EQ_n) \leq m$ if and only if $M_m$ contains a permutation submatrix of size $2^n$.
\end{lemma}
\begin{proof} 
    ($\Rightarrow$) $\GH(\EQ_n) \leq m$ means $\EQ_n$ can be computed by $m$ pipes. In the configuration matrix $M_m$, the intersection of Alice's configurations $\{A(x)|x\in \{0,1\}^n\}$ and Bob's configurations $\{B(y)|y\in \{0,1\}^n\}$ is a permutation submatrix, because the entry $(A(x),B(y))$ is 1 if and only if $x=y$.
    
    ($\Leftarrow$) We take the permutation submatrix of size $2^n$. Then, we label the rows by $\{0,1\}^n$ and label the columns by a permutation of $\{0,1\}^n$, such that the intersection of the row labeled by $x$ and the column labeled by $y$ is 1 if and only if $x=y$. Finally, let $A(x)$ be the configuration indicated by the row labeled by $x$, and $B(y)$ be the configuration indicated by the column labeled by $y$. As a result, $\EQ_n$ can be computed by $m$ pipes.
\end{proof}

We already gain some information from above $M_3$ and $M_4$: 
$$\GH(\EQ_1) = 3,$$ $$\GH(\EQ_2) = 4,$$ $$\GH(\EQ_3) \geq 5.$$

To make the configuration matrix smaller, lemma~\ref{lem:lastrowblock} says we can assume Alice has only one open pipe for even $m$.  Before the lemma, we define some notations first.

For a configuration $A(x)$, we call the pipe which connects the water tap (number 0) the \emph{water-in pipe}, and we call the pipe where the water comes out (from Alice's side) in the configuration $(A(x),B(x))$ the \emph{water-out pipe}.

We divide $M_m$ into blocks 
$M_m = [M_m^{i,j}]_{1 \leq i \leq m/2, 1 \leq j \leq (m-1)/2}$ 
where $M_m^{i,j}$ consists of those rows where Alice has $i$ hoses and those columns where Bob has $j$ hoses. For example, 
$M_3 = M_3^{1,1}, M_4 = 
\begin{bmatrix} 
 M_4^{1,1} \\ 
M_4^{2,1}
\end{bmatrix}$. 

\begin{lemma}
\label{lem:lastrowblock}
For even $m \geq 4$, if $M_m$ contains a permutation submatrix of size $k$, then the last-row-block of $M_m$, 
$[M_m^{m/2,1},M_m^{m/2,2},\cdots,M_m^{m/2,(m-2)/2}]$, also contains a permutation submatrix of size $k$.
\end{lemma}
\begin{proof}
    We take the permutation submatrix of size $k$, denoted by $S\times T$. If $S$ is not contained in the last-row-block, then we take a row $A(x)\in S$, which is outside of the last-row-block. In configuration $A(x)$, Alice has $m$ pipes: a water-in pipe, a water-out pipe, some pairs of pipes connected with hoses, and others. Because $m$ is even, there are even ``other'' pipes. We connect them by hoses arbitrarily. After that, every pipe is connected with a hose, except the water-out pipe. Therefore, this configuration, denoted by $A'(x)$, is in the last-row-block. We replace $A(x)$ by $A'(x)$, and we will show that it is still a permutation submatrix.
    
    First, we will show that water comes out from Alice's side when $A'(x)$ meets $B(x)$. The path water flows in configuration $(A'(x),B(x))$ is the same as the one in configuration $(A(x),B(x))$ and the water comes out from the same pipe and same side, because the new hoses we added are not filled with water at all.
    
    Second, we will show that water comes out from Bob's side when $A'(x)$ meets $B(y)$ if $x\neq y$. Again, the water path did not change, because we did not remove Alice's hoses or add Bob's hoses.
    
    If we repeat the replacement many times, we can move $S\times T$ into the last-row-block.
    
\end{proof}

By the lemma above, we know that if $m$ is even, in order to search maximum size of permutation submatrix in $M_m$, we can restrict ourselves to the last-row-block of $M_m$, which means on Alice's side there is only one pipe without any hose connection. In the following construction, we assume that $m$ is always even.

Actually we do not have to restrict ourselves to permutation matrices of size power two for the following lemma. 

\begin{lemma}
If $M_m$ contains a permutation submatrix of size $k$, then $M_{mt}$ contains a permutation submatrix of size $k^t$ for every $t=1,2,\cdots$. 
\end{lemma}
\begin{proof}
    In $M_m$, we denote the rows and columns of the permutation submatrix of size $k$ by $\{A(x)|x\in [k]\}$ and $\{B(y)|y\in [k]\}$, respectively, where $A(x)$ and $B(y)$ intersect at entry 1 if and only if $x=y$. In $M_{mt}$, we will show that $\{A'(x')|x'\in [k]^t\} \times \{B'(y')|y'\in [k]^t\}$ is permutation submatrix of size $k^t$.
    
    First, we define $A'$ and $B'$. We group $mt$ pipes into $t$ blocks, where each block has $m$ pipes.
    
    $B'$ is just the product of $B$ in each block, i.e. $$B'(y')=\{ \{im+a,im+b\}|\{a,b\}\in B(y'_i),i\in [t] \}.$$ The construction of $A'$ is almost the same as $B'$, but we connect the water-out pipe to the water-in pipe in the next block. Formally speaking, 
    \begin{eqnarray*}
        A'(x')= &     & \{\{im+a,im+b\}|\{a,b\}\in B(y'_i),\ a,b\neq 0,\ i\in [t] \} \\
                &\cup & \{\{0,\waterin(A(x'_0))\}\} \\
                &\cup & \{\{im+\waterout(A(x_i)),(i+1)m+\waterin(A(x_{i+1}))\} | i=1,2,\cdots,t-1 \}.
    \end{eqnarray*}
    
    Second, we show this construction is a permutation submatrix. 
    \begin{itemize}
        \item $x'=y' \Rightarrow M_{mt}[A'(x'),B'(y')]=1$. For block $i$ ($i\in [t]$), water comes into the water-in pipe of block $i$ (number $im+\waterin(A(x'_i))$), and goes out from the water-out pipe (number $im+\waterout(A(x_i))$). Then, it flows to the water-in pipe of the next block. Finally, it goes out from the water-out pipe of the last block from Alice's side.
        \item $x'\neq y' \Rightarrow M_{mt}[A'(x'),B'(y')]=0$. Let $i$ be the minimum number that $x'_i\neq y'_i$. Similar to the case $x'=y'$, water flows thought the same path in blocks $0,1,\cdots,i-1$, and arrives at the water-in pipe of block $i$. Since $x'_i\neq y'_i$, water goes out from Bob's side in block $i$.
    \end{itemize}

\end{proof}

By the two lemmas above, we have the following corollary.
\begin{corollary}
If there exist $m$ and $k$ such that $M_m$ contains a permutation submatrix of size $k$, then the 
garden hose complexity of the equality: $$\GH(\EQ_n) \leq \frac{m}{\log k} n + O(1). $$ 
\end{corollary}

For example when $m=4$, $k=6$ thus $\frac{m}{\log k} = \frac{4}{\log 6} \approx 1.547 $. 
In \cite{MM12} it is shown that when $m=8$, $k = 46$ thus $\frac{m}{\log k} = \frac{8}{\log 46} \approx 1.448 $ which is the best known result before this paper.   

For larger $m$, we are devising computer programs to search for the maximum size of permutation submatrix in $M_m$. We need more properties to search faster.

\begin{lemma}
    Bob's one configuration cannot cover another one, i.e. $B(y)\subseteq B(y') \Rightarrow y=y'$.
\end{lemma}
\begin{proof}
    The water comes out from Alice's side in configuration $(A(y),B(y))$. If we add more hoses on Bob's side, the water path does not change. So, the water comes out from the same pipe on Alice's side in configuration $(A(y),B(y'))$. Therefore, $y=y'$.
\end{proof}

To find a large permutation submatrix, we first want to find a large set of Bob's configurations such that one doesn't cover another. In the program, we assume each configuration in the set has the same number of pipes, so we don't need to worry about the covering problem. Moreover, to maximize the size of the set, we assume Bob's hoses covers almost half of the pipes. Appendix \ref{sec:saprogram} gives more ideas of our simulated annealing program.

The following table are results of computer program. 


\begin {table} [h]
\begin{tabular}{c||c|c|c|c|c}
$m$ & 4 & 6 & 8 & 10 & 12 \\
\hline
$k$ & 6 & 15 & 48 & 144 & 395 \\
\hline 
running time & $<1$ sec & $<1$ sec & $10$ sec & $1$ hour & $1$ day \\
\hline
$m/\log k$ & 1.547... & 1.535... & 1.432... & 1.394... & 1.391... \\ 
\end{tabular}
\caption{Search results on even $m\le12$. \label{P1}}
\end {table}



\section{Some symmetric solutions} \label{symsol}

From previous section we know that when $m=12$, $k=395$, $\frac{m}{\log k} = \frac{12}{\log 395} \approx 1.391$ which is already better than 1.448 in \cite{MM12}. However, if we want better, unstructured searching program does not help any more. One possible hope is to study the structure of solutions, in particular we are interested in those symmetries 
behind some solutions. In this section we use one example to explain the basic idea while generalize in next. 

It is known that $M_4$ contains a permutation submatrix of size 6. Here we write down one solution(the permutation submatrix). 

$$\begin{matrix}
 & \begin{matrix}
\text{ 12 } & \text{13} & \text{14 } & \text{23} & \text{24 } & \text{34 }
\end{matrix} \\ 
\begin{matrix}
01,23\\
01,24\\
01,34\\
02,13\\
02,14\\
03,12
\end{matrix} &  \begin{bmatrix}
\text{ 0 } & \text{ 0 } & \text{ \textcolor{red}{1} } & \text{ 0 } & \text{ 0 } & \text{ 0 }\\ 
\text{ 0 } & \text{ \textcolor{red}{1} } & \text{ 0 } & \text{ 0 } & \text{ 0 } & \text{ 0 }\\ 
\text{ \textcolor{red}{1} } & \text{ 0 } & \text{ 0 } & \text{ 0 } & \text{ 0 } & \text{ 0 }\\ 
\text{ 0 } & \text{ 0 } & \text{ 0 } & \text{ 0 } & \text{ \textcolor{red}{1} } & \text{ 0 }\\ 
\text{ 0 } & \text{ 0 } & \text{ 0 } & \text{ \textcolor{red}{1} } & \text{ 0 } & \text{ 0 }\\  
\text{ 0 } & \text{ 0 } & \text{ 0 } & \text{ 0 } & \text{ 0 } & \text{ \textcolor{red}{1} }
\end{bmatrix}
\end{matrix}$$

Here Alice and Bob are computing the equality function $f(x,y)$ where $x,y \in [6]$. On input $x$, Alice makes a configuration $A(x)$, while 
on input $y$, Bob makes a configuration $B(y)$. Thus we can also write the above solution in the following form. 
$$A(1) = \{01,23\}, B(1) = \{14\} $$ 
$$A(2) = \{01,24\}, B(2) = \{13\} $$ 
$$A(3) = \{01,34\}, B(3) = \{12\} $$
$$A(4) = \{02,13\}, B(4) = \{24\} $$ 
$$A(5) = \{02,14\}, B(5) = \{23\} $$ 
$$A(6) = \{03,12\}, B(6) = \{34\} $$

The symmetry of the above solution can be seen in at least two ways, although equivalent. 

First let us assume that Alice and Bob only know how to connect pipes on input $(1,1)$, i.e. they only have $$A(1) = \{01,23\}, B(1) = \{14\} $$ now the question is: how can they
get other configurations from $(A(1), B(1))$? One possible way is to use a group to act on $(A(1), B(1))$. By doing this, hopefully they can get all other configurations.
In other words, the final solution is invariant under some group action. 

Two groups fit for our setting here, the alternating group $A_4$ and the symmetric group $S_4$. Take $A_4$ for example. 
$$A_4 = \{(1), (123), (132), (124), (142), (134), (143), (234), (243), (12)(34), (13)(24), (14)(23)\}$$
For an element $\sigma \in A_4$, define its action on $A(1), B(1)$ as $$A(\sigma) = \{0\sigma(1),\sigma(2)\sigma(3)\}, B(\sigma) = \{\sigma(1)\sigma(4)\} $$ 
then we have 
$$\begin{matrix}
A((1)) = \{01,23\}, B((1)) = \{14\} & \color{red}{\rightarrow} & A(1), B(1) \\
A((123)) = \{02,13\}, B((123)) = \{24\} & \color{red}{\rightarrow} & A(4), B(4) \\
A((132)) = \{03,12\}, B((132)) = \{34\} & \color{red}{\rightarrow} & A(6), B(6) \\
A((124)) = \{02,34\}, B((124)) = \{12\} & & \\
A((142)) = \{04,13\}, B((142)) = \{24\} & & \\
A((134)) = \{03,24\}, B((134)) = \{13\} & & \\
A((143)) = \{04,12\}, B((143)) = \{34\} & & \\
A((234)) = \{01,34\}, B((234)) = \{12\} & \color{red}{\rightarrow} & A(3), B(3)  \\
A((243)) = \{01,24\}, B((243)) = \{13\} & \color{red}{\rightarrow} & A(2), B(2) \\
A((12)(34)) = \{02,14\}, B((12)(34)) = \{23\} & \color{red}{\rightarrow} & A(5), B(5) \\
A((13)(24)) = \{03,14\}, B((13)(24)) = \{23\} & & \\
A((14)(23)) = \{04,23\}, B((14)(24)) = \{14\} & &
\end{matrix}$$

Note that the group action also generates other configurations which are not in the solution, for example $A((124)) = \{02,34\}$. However, this should not be a problem since 
if we look into the row of $A((124))$ in matrix $M_4$, it is exactly the same as row $A((234)) = \{01,34\}$, 
thus we still get a solution by choosing any one of $A((234))$ or $A((124))$. 



Another way to look into the symmetry of the above solution is through geometry. Treat four pipes as four vertices of a regular tetrahedron. Let vertex $1$ be water-in 
pipe and vertex $4$ be water-out pipe, 2 and 3 are connected by Alice while 1 and 4 are connected by Bob. This is the initial configuration $(A(1), B(1))$ as before, 
and call it base construction. Now rotate or reflect the tetrahedron, then the rotation group $A_4$ or the symmetric group $S_4$ of the tetrahedron are exactly 
the same as we discussed from the first perspective.  Also the meaning of equivalence between $A((234))$ and $A((124))$ discussed above can be explained 
here as we are not making any difference by switching the water-in and water-out pipe since it does not change the side where water comes out.

\begin{figure}[htb]
\begin{center}
\includegraphics[scale=0.7]{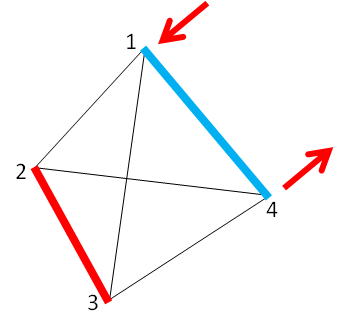}
\caption{ A geometric explanation of symmetric solutions }
\end{center}
\end{figure}

\section{Group invariant framework}

The symmetry discussion in previous section can be summarized by the following formula:
$$\fbox{\text{ symmetric solution } = \text{ base construction } + \text{group action} } $$
Now we generalize this idea here. 

Let $m$ be an even integer as before. Let $t$ be the number of hoses that Bob uses. Without loss of generality we assume that on input $(1,1)$, 
$$A(1) = \{01,23,\cdots,(2t-2)(2t-1), (2t+1)(2t+2), \cdots, (m-1)m\}$$ 
$$B(1) = \{12, 34, \cdots, (2t-1)(2t)\}$$
Here $(A(1), B(1))$ is the \textit{base construction}. 

Let $G \leq S_m$ be a permutation group. For every $g \in G$, define the \textit{group action} as 
$$A(g) = \{0g(1),g(2)g(3),\cdots,g(2t-2)g(2t-1), g(2t+1)g(2t+2), \cdots, g(m-1)g(m)\}$$
\begin{equation}
 B(g) = \{g(1)g(2), g(3)g(4), \cdots, g(2t-1)g(2t)\}. \label{gaction}
\end{equation}

Now we can check matrix $M_m$ if the intersection of rows $\{A(g) | g \in G \}$ and columns $\{B(g) | g \in G \}$ forms a permutation submatrix. If so, then we call $G$ an $(m,t)$-\textbf{garden hose permutation group}. 
However, in some cases, the intersection may contain repeated rows or columns. If we remove those repeated ones 
and the rest still forms a permutation submatrix, then we call $G$ a \textbf{weak} $(m,t)$-garden hose permutation group. Correspondingly, call the solution 
$\{A(g), B(g) | g \in G \}$ a (weak) group invariant solution. For example, $S_4$ and $A_4$ discussed in section \ref{symsol} are both $(4,1)$-weak garden hose permutation groups.

Now we are reducing the problem of the garden hose complexity for equality function to the problem of deciding which group is a garden hose permutation group. 
If there is an $(m,t)$-garden hose permutation group $G$, then the garden hose complexity for equality function 

\begin{equation} \label{ghgratio}
 \GH(\EQ_n) \leq \frac{m}{\log |G|} n + O(1).
\end{equation}

\section{Group invariant construction} \label{gicon}

For a given $m$, we can check all conjugacy classes of the subgroups of symmetric group $S_m$, to see which one is a (weak) garden hose permutation group. Each of them gives an upper bound for $\GH(\EQ_n)$
according to formula (\ref{ghgratio}). Tables of Marks~\cite{tablesofmarks} shows conjugacy classes of the subgroups of small symmetric groups. Note that to check if a submatrix in $M_m$ is a permutation matrix is much easier than searching a permutation matrix in it. Table \ref{ghglist} lists some results. 

\begin{figure}[htb] \label{ghglist}
\begin{center}
\includegraphics[scale=0.5]{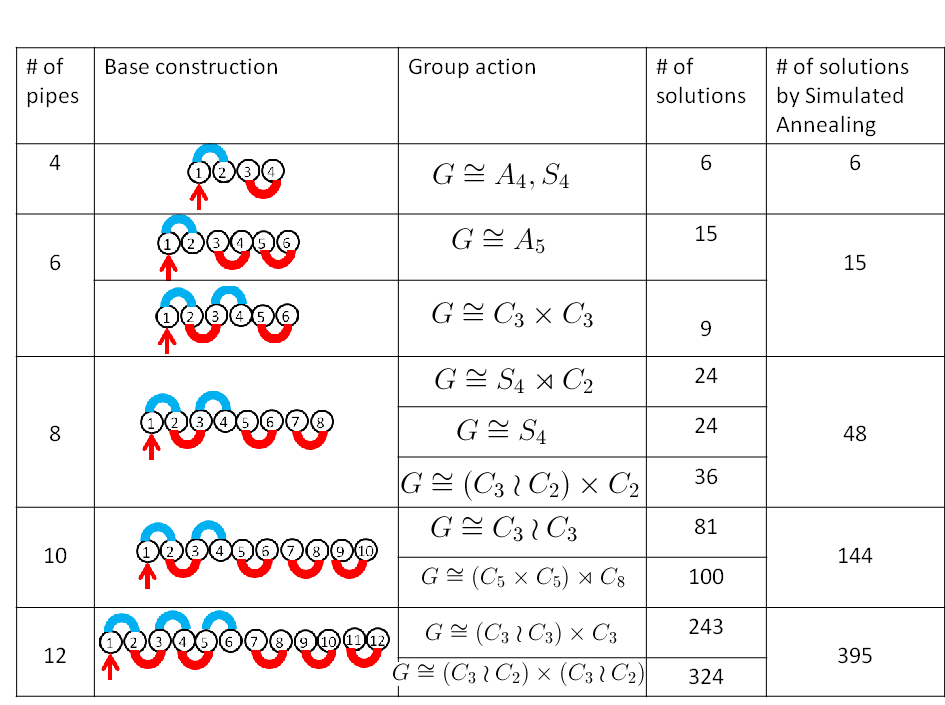}
\caption{Lists of some (weak) garden hose permutation groups }
\end{center}
\end{figure}

\subsection{Wreath product construction}
There is a particular group which stands out in our search results. Denote it by $W_2$. $W_2(\leq S_{10})$ has generators 
$$( 3, 5, 10), ( 2, 7, 8), ( 1, 2, 3)( 4, 7, 10)( 5, 6, 8).$$ For those familiar with wreath product\cite{OOR04}, we know that $$W_2 \cong C_3 \wr C_3, $$
the wreath product of two cyclic group of order $3$.
Computer program tells us that $W_2$ is a $(10,2)$ garden hose permutation group. Since $|W_2|=3^4=81$, by formula (\ref{ghgratio}) we get 
$\frac{m}{\log |W_2|} = \frac{10}{\log 81} \approx 1.577$. This ratio is worse than best result by our simulated annealing program, 
and it is even worse than 1.448 by Margalit and Matsliah\cite{MM12}. However, the reason that why it stands out from others is its 
special structure: the wreath product, from which there is a way we can generalize, in hope of beating best current result.

A natural question is: is there any large $l$ such that there exists a garden hose permutation group $W_l$ and 
$W_l \cong  \underset{l}{\underbrace{C_3 \wr \cdots \wr C_3}}$? We conjecture that this is the case. 

\begin{conjecture} [3-tree conjecture] \label{3treeconj}
 For every $l \geq 2$ there exists a group $W_l \leq S_{3^l+1}$, $W_l \cong  \underset{l}{\underbrace{C_3 \wr \cdots \wr C_3}} $ and
 $W_l$ is a $(3^l+1, t)$ garden hose permutation group for some $t$. 
\end{conjecture}

One consequence of the 3-tree conjecture is: we can push the upper bound for $\GH(\EQ_n)$ even further. 

\begin{theorem} \label{3treeconseq}
 If conjecture \ref{3treeconj} holds, then the garden hose complexity of the equality function: $$\GH(\EQ_n) \leq \frac{2}{\log 3} \cdot n + O(1) \approx 1.262 n + O(1). $$
\end{theorem}
\begin{proof}
 Since $$|W_l| = 3^{1+3+ \cdots + 3^{l-1}} = 3^{\frac{3^l-1}{2}},$$ and $m = 3^l + 1$, thus by formula (\ref{ghgratio}), we have 
$$\frac{m}{\log |W_l|} = \frac{3^l+1}{\log (3^{\frac{3^l-1}{2}})} = \frac{2}{\log 3} \cdot \frac{3^l+1}{3^l-1},$$ for large enough $l$,
the above tends to $\frac{2}{\log 3} \approx 1.262$.
\end{proof}

We already know that the 3-tree conjecture holds for $l = 2$, actually it also holds for $l = 3$ which gives us better upper bound for $\GH(\EQ_n)$ unconditionally. 

\begin{proof} [proof of theorem \ref{mainthm}]
It suffices to provide conjugate representatives for $W_3 = g^{-1}Kg \cong C_3 \wr C_3 \wr C_3$ where \\
$g= (3,17,4,15)(6,24,20,11)(2,28,18,5,26,22,14,2)(7,21,13,8,25,23,19,10,16)$
and $K$ is generated by: \\
$(1,2,3), (4,5,6), (1,4,7)(2,5,8)(3,6,9),$ \\
$(10,11,12), (13,14,15), (10,13,16)(11,14,17)(12,15,18),$ \\
$(19,20,21), (22,23,24), (19,22,25)(20,23,26)(21,24,27),$\\
$(1,10,19)(2,11,20)\cdots(9,18,27)$ \\
Then, by computer programs, one can check it is a $(28, 7)$ garden hose permutation group. Thus $\frac{28}{\log 3^{13}} \approx 1.359.$
\end{proof}

\section{Further research}

At least two questions arise. First try to prove the 3-tree conjecture. Second, is there any other garden hose permutation group besides $W_2$ and $W_3$ 
such that the structure can be generalized in order to get better upper bound for $\GH(\EQ_n)$? 

The authors would like to thank Mike Saks for useful discussions. 

\bibliographystyle{alpha}
\bibliography{gh}

\newcommand{\etalchar}[1]{$^{#1}$}
\begin{thebibliography}{CGMO09}

\bibitem[BCF{\etalchar{+}}11]{Position11Imp}
Harry Buhrman, Nishanth Chandran, Serge Fehr, Ran Gelles, Vipul Goyal, Rafail
  Ostrovsky, and Christian Schaffner.
\newblock Position-based quantum cryptography: Impossibility and constructions.
\newblock In {\em CRYPTO}, pages 429--446, 2011.

\bibitem[BFS11]{Position11}
Harry Buhrman, Serge Fehr, and Christian Schaffner.
\newblock Position-based quantum cryptography.
\newblock {\em ERCIM News}, 2011(85):16--17, 2011.

\bibitem[BFSS13]{BFSS13}
Harry Buhrman, Serge Fehr, Christian Schaffner, and Florian Speelman.
\newblock The garden-hose model.
\newblock In {\em ITCS}, pages 145--158, 2013.

\bibitem[CGMO09]{CGMO09}
Nishanth Chandran, Vipul Goyal, Ryan Moriarty, and Rafail Ostrovsky.
\newblock Position based cryptography.
\newblock In {\em CRYPTO}, pages 391--407, 2009.

\bibitem[MM12]{MM12}
O.~Margalit and A.~Matsliah.
\newblock Mage - the cdcl sat solver developed and used by ibm for formal
  verification \url{http://ibm.co/P7qNpC}.
\newblock {\em personal communication}, 2012.

\bibitem[OOR04]{OOR04}
R.C. Orellana, M.E. Orrison, and D.N. Rockmore.
\newblock Rooted trees and iterated wreath products of cyclic groups.
\newblock {\em Advances in Applied Mathematics}, 33(3):531 -- 547, 2004.

\bibitem[PM]{tablesofmarks}
G{\"o}tz Pfeiffer and Thomas Merkwitz.
\newblock {GAP Data Library ``Tables of Marks''}.
\newblock \url{http://www.gap-system.org/Datalib/tom.html}.

\end{thebibliography}

\begin{appendix}
 
\section{Simulated Annealing Program}
\label{sec:saprogram}
This program searches for a large permutation submatrix of the configuration matrix $M_m$. 

We assume $m$ is even, Alice's configuration has only one open pipe, and Bob's hoses covers half of the pipes.

First we find a small permutation submatrix, than it grows larger. In each iteration, we try to add one row or one column. If it fails to grow at some point, we remove one row and/or one column from it, and try again. We give the pseudo-code as below.

\lstset{
language=Java,
basicstyle=\footnotesize,
frame=single,
morekeywords={with, prob, and, print}
}
\begin{lstlisting}
/*
 * X is a set of rows. Y is a set of columns. 
 * The submatrix X * Y is the permutation submatrix we want to find. 
 */ 
X = empty set
Y = empty set

add_a_row:
for x in (all rows \ X) in random order {
    if submatrix {x} * Y is all-zero {
        goto add_a_column
    }
}
/* Failed to add a row. Remove a row and a column. */
Pick a 1 in submatrix X * Y at random. 
Denote the location of the 1 by (x1,y1).
Remove x1 from X.
Remove y1 form Y.
goto add_a_row

add_a_column: 
for y in (all columns \ Y) in random order {
    if submatrix X * {y} is all-zero and (x,y) is one {
        /* We have a larger one. */
        X = X union {x}
        Y = Y union {y}
        print X * Y is a permutation submatrix
        goto add_a_row
    }
}
/* Failed to add a row. */
with prob. 1/4 {
    goto add_a_row /* Discard y. */
}
/* with prob. 3/4 */
/* Remove a row and a column. */
Pick a 1 in submatrix X * Y at random. 
Denote the location of the 1 by (x1,y1).
Remove x1 from X.
Remove y1 form Y.
goto add_a_column
\end{lstlisting}

We run the program on a PC. If it has not found a larger permutation submatrix (no ``print'') in a long time, we kill the program and run it again. The following table shows the search results on $m\leq 12$, where $k$ is the size of the maximum permutation submatrix the program find.

\begin{tabular}{c||c|c|c|c|c}
$m$ & 4 & 6 & 8 & 10 & 12 \\
\hline
$k$ & 6 & 15 & 48 & 144 & 395 \\
\hline 
running time & $<1$ sec & $<1$ sec & $10$ sec & $1$ hour & $1$ day \\
\hline
$m/\log k$ & 1.547... & 1.535... & 1.432... & 1.394... & 1.391... \\ 
\end{tabular}

\end{appendix}

\end{document}